\documentclass[a4paper,11pt]{article}

	\usepackage[english,francais]{babel}		
	\usepackage[T1]{fontenc}			
	\usepackage[utf8x]{inputenc}
	\usepackage{amsmath,amsfonts,amssymb,amsthm}	
	\usepackage{xspace}				

	\def\DocDate{\today}
	\def\Authors{Bruno Grenet%
            \thanks{Electronic mail address: bruno.grenet@ens-lyon.fr.}\\[2pt] 
            Laboratoire de l'Informatique du Parall\'elisme,\\
            \'Ecole Normale sup\'erieure de Lyon,\\
            46, all\'ee d'Italie,\\ 
            69\,364 Lyon Cedex 07, France%
        }
	\def\LongTitle{Acceptable Complexity Measures of Theorems}

        \usepackage{fullpage}
	\usepackage{enumerate}				

        \numberwithin{equation}{section}

        \author{\Authors}
        \date{\DocDate}
        \title{\LongTitle}

        \makeatletter
        \def\@biblabel#1{#1.}
        \makeatother

        \newenvironment{enumeration}
        {\begin{enumerate}[\itshape (i) \upshape]}
        {\end{enumerate}}

	\theoremstyle{plain}
	\newtheorem{theorem}{Theorem}[section]
	\newtheorem{lemma}[theorem]{Lemma}
	\newtheorem{proposition}[theorem]{Proposition}
        \newtheorem{corollary}[theorem]{Corollary}
	\theoremstyle{definition}
	\newtheorem{definition}[theorem]{Definition}  
	
        \newtheorem{example}[theorem]{Example}
	\theoremstyle{remark}

        \numberwithin{equation}{section}

\newcommand{\N}{\mathbb{N}}
\newcommand{\Q}{\mathbb{Q}}

\newcommand\ensemble[2]{\left\{#1:\vphantom{#1}#2\right\}}

\usepackage{stmaryrd}

\newcommand{\abs}[1]{\left|#1\right|}

\renewcommand{\phi}{\varphi}

\newcommand\prog[1]{\ensuremath{\textit{PROG}_{U_{#1}} }}
\newcommand\ceil[1]{\left\lceil #1 \right\rceil}
\newcommand\floor[1]{\left\lfloor #1 \right\rfloor}
\newcommand\ceillog[2]{\ensuremath{\ceil{\log_{#1} #2}} }
\newcommand\T{\mathcal{T}\xspace{}}
\newcommand\F{\mathcal{F}\xspace{}}
\renewcommand\O{\mathcal{O}}
\newcommand\card{\textrm{card}}

\newcommand\cdt[1]{(\ref{cdt-#1})}
\newcommand\textand{\text{ and }}
\newcommand\textif{\text{ if }}
\newcommand\textelse{\text{else}}
\renewcommand\epsilon\varepsilon

\renewcommand\emptyset\varnothing

\begin{document}
\selectlanguage{english}

\maketitle
\abstract{
In 1930, G\"odel \cite{godel1931} presented in K\"onigsberg his famous Incompleteness Theorem, stating that some true mathematical statements are 
unprovable. Yet,
this result gives us no idea about those \emph{independent} (that is, true and unprovable) statements, about their frequency, the reason they are unprovable,
and so on. Calude and J\"urgensen \cite{csi05} proved in 2005 Chaitin's ``heuristic principle'' for an appropriate measure: \emph{the theorems of a finitely-specified
theory cannot be significantly more complex than the theory itself} (see \cite{chaitin1974}). In this work, we investigate the existence of other measures, different from the 
original one, which satisfy this ``heuristic principle''. At this end, we introduce the definition of \emph{acceptable complexity measure of theorems}.
}

\sloppy
\section{Introduction}\label{sec-intro} 

In 1931, G\"odel \cite{godel1931} presented in K\"onigsberg his famous \emph{(first) Incompleteness Theorem}, stating that some true mathematical statements 
are unprovable. More formally and in modern terms, it states the following: 
\begin{quote}
Every computably enumerable, consistent axiomatic system containing elementary arithmetic is incomplete, that is, there exist true sentences
unprovable by the system.
\end{quote}

The truth is here defined by the standard model of the theory we consider. Yet,
this result gives us no idea about those \emph{independent} (that is, true and unprovable) statements, about their frequency, the reason they are unprovable,
and so on. Those questions of quantitative results about the independent statements have been investigated by Chaitin \cite{chaitin1974} in a first time,
and then by Calude, J\"urgensen and Zimand \cite{calude1994} and Calude and J\"urgensen \cite{csi05}. A state of the art is given in \cite{ipp08}. 
Those results state that in both topological and probabilistic terms, incompleteness is a widespread phenomenon. Indeed, unprovability appears as the
norm for true statements while provability appears to be rare. This interesting result brings two more questions. Which true statements are provable, and
why are they provable when other ones are unprovable?

Chaitin \cite{chaitin1974} proposed an ``heuristic principle'' to answer the second question: \emph{the theorems of a 
finitely-specified theory cannot be significantly more complex than the theory itself}. It was proven \cite{csi05} that Chaitin's ``heuristic principle''
is valid for an appropriate measure. This measure is based on the program-size complexity: The complexity $H(s)$ of a binary string $s$ is the length
of the shortest program for a self-delimiting Turing machine (to be defined in the next section) to calculate $s$ (see 
\cite{kolmo1968,chaitin1975,irap02,li1997}). We consider the following computable variation of the program-size complexity:
\[\delta(x)=H(x)-\abs x.\]

This measure gives us some indications about the reasons of unprovability of certain statements. It would be very
interesting to have other results in order to understand the Incompleteness Theorem. Among them, one can try to prove a kind of reverse of the 
theorem Calude and J\"urgensen proved. Their theorem states that there exists a constant $N$ such that any theory which satisfies the hypothesis
of G\"odel's Theorem cannot prove any statements $x$ with $\delta(x)>N$. Another question of interest could be the following: Does there exist
any independent statements with a low $\delta$-complexity? 

Those results are only examples of what can be investigated in this domain.
Yet, such results seem to be hard to prove with the $\delta$-complexity. The aim of our work is to find other complexities which satisfy this 
``heuristic principle'' in order to be able to prove the remaining results. At this end, we introduce the notion of \emph{acceptable complexity
measure of theorems} which captures the important properties of $\delta$. After studying the results of \cite{csi05} about $\delta$, we 
define the acceptable complexity measures. We study their properties, and try to find some other acceptable complexity measures, different from
$\delta$.

The paper is organized as follows. We begin in Section \ref{sec-prerequisites} by some notations and useful definitions. 
In Section \ref{sec-delta}, we present the results of \cite{csi05} with some corrections. Section \ref{sec-acm}
is devoted to the definition of the \emph{acceptable complexity measure of theorems}, and some counter-examples will be given in Section \ref{sec-indep}.
This section is also devoted to the proof of the independence of the conditions we impose on a complexity to be acceptable.
In Section \ref{sec-form}, we will be interested in the possible forms of those acceptable complexity measures.

\section{Prerequisites and notations}\label{sec-prerequisites} 

In the sequel, $\N$ and $\Q$ respectively denote the sets of natural integers and rational numbers. For an integer $i\ge2$, $\log_i$ is the base
$i$ logarithm. We use the notations $\floor\alpha$ and $\ceil\alpha$ respectively for the floor and the ceiling of a real $\alpha$. The cardinality
of a set $S$ is denoted by $\card(S)$. For every integer $i\ge2$, we fix an alphabet $X_i$ with $i$ elements, $X_i^*$ being the set of finite strings
on $X_i$, including the empty string $\lambda$, and $\abs w_i$ the length of the string $w\in X_i$.

We assume the reader is familiar with Turing machines processing strings \cite{turing1936} and with the basic notions of computability theory (see,
for example \cite{sipser,papadimitriou,odifreddi}). We recall that a set is said computably enumerable (abbreviated c.e.) if it is
the domain of a Turing machine, or equivalently if it can be algorithmically listed.

The complexity measures we study are \emph{computable variation} of the \emph{program-size complexity}. In order to define it, we define the
\emph{self-delimiting Turing machines}, shortly \emph{machines}, which are Turing machines the domain of which is a prefix-free set. 
A set $S\subset X^*_i$ is said \emph{prefix-free} if no string of $S$ is a proper extension of another one. In other words, if $x,y\in S$ and
if there exists $z$ such that $y=xz$, then $z=\lambda$. We denote by 
$\textit{PROG}_T=\ensemble{x\in X^*_i}{T\text{ halts on }x}$ the program set of the Turing machine $T$. 
We recall two important results on prefix-free sets. If $S\subset X_i^*$ is
a prefix-free set, then Kraft's Inequality holds: $\sum_{k=1}^\infty r_k\cdot i^{-k}\le1$, where $r_k=\ensemble{x\in S}{\abs x_i=k}$. The second result
is called the Kraft-Chaitin Theorem and states the following: Let $(n_k)_{k\in\N}$ be a computable sequence of non-negative integers such that
\[\sum_{k=1}^\infty i^{-n_k}\le1,\]
then we can effectively construct a prefix-free sequence of strings $(w_k)_{k\in\N}$ such that for each $k\ge1$, $\abs{w_k}_i=n_k$.

The \emph{program-size complexity} of a string $x\in X^*_Q$, relative to the machine $T$, is defined by
\[H_{i,T}=\min\ensemble{\abs y_i}{y\in X^*_i\textand T(y)=x}.\]
In this definition, we assume that $\min(\emptyset)=\infty$. The Invariance Theorem ensures the effective existence of a so-called \emph{universal}
machine $U_i$ which minimize the program-size complexity of the strings. For every $T$, there exists a constant $c>0$ such that for all 
$x\in X^*_i$, $H_{i,U_i}(x)\le H_{i,T}(x)+c$. In the
sequel, we will fix $U_i$ and denote by $H_i$ the complexity $H_{i,U_i}$ relative to $U_i$.

A \emph{G\"odel numbering} for a formal language $L\subseteq X_i^*$ is a computable, one-to-one function $g:L\to X_2^*$. By $G_i$, or $G$ if there is no 
possible confusion, we denote the set of all the G\"odel numbering for a fixed language. In what follows, we consider theories which satisfy the
hypothesis of G\"odel Incompleteness Theorem, that is finitely-specified, sound and consistent theories strong enough to formalize arithmetic. The first
condition means that the set of axioms of the theory is c.e.; soundness is the property that the theory only proves true sentences; consistency states that
the theory is free of contradictions.
We will generally denote by $\F$ such a theory, and by $\T$ the set of theorems that $\F$ proves.

\section{The function $\delta_g$}\label{sec-delta} 

We present in this section the function $\delta_g$ and some results about it. It was defined in \cite{csi05} and almost all the results come
from this paper. Hence, complete proofs of the results can be found in it. Yet, there was a mistake in the paper, and we need to
modify a bit the definition of $\delta_g$. We have to adapt the proofs with the new definition. The transformations are essentially cosmetic
in almost all the proofs so we give only sketches of them. For Theorem \ref{thm-inva}, there are a bit more than details to change,
so we provide a complete proof of this result. Furthermore, we formally prove an assertion used in the proof of Theorem \ref{thm-incompl}.

We first define, for every integer $i\ge 2$, the function $\delta_i$ by
\[\delta_i(x) = H_i(x)-\abs{x}_i.\]
Now, in order to ensure that the complexity we study is not dependent on the way we write the theorems, we define the $\delta$-complexity
\emph{induced by a G\"odel numbering} $g$ by\footnote{The definition in \cite{csi05} was $\delta_g(x)=H_2(g(x))-\ceillog2i\cdot\abs x_i$.}
\[\delta_g(x) = H_2(g(x))-\ceil{\log_2(i)\cdot\abs x_i},\]
where $g$ is a G\"odel numbering the domain of which is in $X^*_i$.

The first result comes in fact from \cite{irap02}, and the theorem we present here is one of its direct corollaries.

\begin{theorem}[\protect{\cite[Corollary 4.3]{csi05}}]\label{thm-inf} 
For every $t\ge 0$, the set $\ensemble{x\in X^*_i}{\delta_i(x)\le t}$ is infinite.
\end{theorem} 

\begin{proof} 
Following \cite[Theorem 5.31]{irap02}, for every $t\ge 0$, the set $C_{i,t}=\ensemble{x\in X^*_i}{\delta_i(x)>-t}$ is immune\footnote{A set
is said immune when it is infinite and contains no infinite c.e. subset.}. Hence, as $\text{Complex}_{i,t}=\ensemble{x\in X^*_i}{\delta_i(x)>t}$
is an infinite subset of an immune set, it is immune itself. The set in the statement being the complement of the immune set $\text{Complex}_{i,t}$,
it is not computable, and in particular infinite.
\end{proof} 

The next theorem states that the definitions \emph{via} a G\"odel numbering or without this device are not far from each other. It allows us to work
with the function $\delta_i$ instead of $\delta_g$ and thus to simplify the proofs thanks to the elimination of some technical details. 
Nevertheless, those details are present in the following proof.

\begin{theorem}[\protect{\cite[Theorem 4.4]{csi05}}]\label{thm-inva} 
Let $A\subseteq X_i^*$ be c.e. and $g:A\to B^*$ be  a G\"odel numbering. Then, there effectively exists a constant $c$ (depending upon 
$U_i, U_2$, and $g$) such that for all $u\in A$ we have
\begin{equation}
\label{eq-inva}
\abs{H_2(g(u))-\log_2(i)\cdot H_i(u)}\le c.
\end{equation}
\end{theorem} 

\begin{proof}
We will in fact prove the existence of two constants $c_1$ and $c_2$ such that on one hand
\begin{equation}\label{eq-c1}
H_2(g(u))\le \log_2(i)\cdot H_i(u)+c_1
\end{equation}
and on the other hand
\begin{equation}\label{eq-c2}
\log_2(i)\cdot H_i(u)\le H_2(g(u))+c_2.
\end{equation}

For each string $w\in\prog i$, we define $n_w=\ceil{\log_2(i)\cdot\abs w_i}$. This integers verify the following:
\[\sum_{w\in\prog i} 2^{-n_w} = \sum_{w\in\prog i} 2^{-\ceil{\log_2(i)\cdot\abs w_i}} \le \sum_{w\in\prog i} i^{-\abs w_i} \le 1,\]
because $\prog i$ is prefix-free. This inequality shows that the sequence $(n_w)$ satisfies the conditions of
the Kraft-Chaitin Theorem. Consequently, we can construct, for every $w\in\prog i$, a binary string $s_w$ of length $n_w$ and such that the
set $\ensemble{s_w}{w\in\prog i}$ is c.e. and prefix-free. Accordingly, we can construct a machine $M$ whose domain is this set, and such that for 
every $w\in\prog i$,
\[M(s_w)=g(U_i(w)).\]
If we denote, for a string $x\in X^*_i$, $x^*$ the lexicographically first string of length $H_i(x)$ such that $U_i(x^*)=x$, we now 
have $M(s_{w^*})=g(U_i(w^*))=g(w)$, and hence
\begin{eqnarray*}
H_M(g(w))\le\abs{s_{w^*}}_2&=&\ceil{\log_2(i)\cdot\abs{w^*}_i}\\
&=&\ceil{\log_2(i)\cdot H_i(w)}\le \log_2(i)\cdot H_i(w)+1.
\end{eqnarray*}
By the Invariance Theorem, we get the constant $c_1$ such that \eqref{eq-c1} holds true.

We now prove the existence of $c_2$ such that \eqref{eq-c2} holds true. The proof is quite similar. For each string $w\in\prog 2$, we define
$m_w=\ceil{\log_i(2)\cdot\abs w_2}$. As for the $n_w$, the integers $m_w$ satisfy
\[\sum_{w\in\prog2} i^{-m_w} \le \sum_{w\in\prog2} 2^{-\abs w_2} \le 1.\]
We can also apply the Kraft-Chaitin Theorem to effectively construct, for every $w\in\prog2$, a string $t_w\in X^*_i$
of length $m_w$ and such that the set $\ensemble{t_w}{w\in\prog2}$ is c.e. and prefix-free. As $g$ is a G\"odel numbering
and hence one-to-one, we can construct a machine $D$ whose domain is the previous set and such that $D(t_w)=u\textif U_2(w)=g(u)$.
Now, if $U_2(w)=g(u)$, then 
\begin{eqnarray*}
H_D(u)\le\ceil{\log_i(2)\cdot\abs w_2}&\le&\log_i(2)\cdot\abs w_2+1\\
&\le&\log_i(2)\cdot H_2(g(u))+d.
\end{eqnarray*}
So we apply the Invariance Theorem to get a constant $d'$ such that $\log_2(i)\cdot H_i(u)\le\log_2(i)\cdot H_D(u)+d'$, hence
\[\log_2(i)\cdot H_i(u)\le H_2(g(u))+d+d'.\]
The constant $c_2=d+d'$ satisfies \eqref{eq-c2}.
\end{proof}

In \cite{csi05}, the equation \eqref{eq-inva} was $\abs{\delta_g(u)-\ceillog2i\cdot\delta_i(u)}\le d$. Theorem \ref{thm-inva} gives a similar result
for $\delta$, hence $\abs{\delta_g(u)-\log_2(i)\cdot\delta_i(u)}\le c+1$, where $c$ is the constant of the theorem. In the proof, we supposed that
$A=X^*_i$ but it is still valid with a proper subset of $X^*_i$.

The next corollary will be important for the generalization of $\delta_g$ we will do in the next section. It is the same kind of result as above,
but applied to two G\"odel numberings.

\begin{corollary}[\protect{\cite[Corollary 4.5]{csi05}}]\label{cor-inva} 
Let $A\subseteq X_i^*$ be c.e. and $g, g' : A \to B^*$ be two  G\"odel numberings. Then, there effectively exists a constant $c$ (depending upon
$U_2, g$ and $ g'$) such that for all $u\in A$ we have:
\begin{equation}
\label{eq-invg}
\abs{H_2(g(u))-H_2(g'(u))}\le c.
\end{equation}
\end{corollary} 

In order to have a complete formal proof of Theorem \ref{thm-incompl}, we need to bound the complexity of the set $\T$ of theorems that a theory $\F$
proves. It is the aim of the following lemma.

\begin{lemma}\label{lem-bound-th} 
Let $\F$ be a finitely-specified, arithmetically sound (i.e. each arithmetical proven sentence is true), consistent theory strong enough to formalize
arithmetic, and denote by $\T$ its set of theorems written in the alphabet $X_i$. Then for every $x\in\T$,
\[\frac12\cdot\abs x_i+\O(1)\le H_i(x)\le\abs{x}_i+\O(1).\]
\end{lemma} 

\begin{proof} 
For the upper bound, it is sufficient to give a way to describe those theorems using descriptions not greater than their lengths, and which
ensure that the computer we use is self-delimiting. 
We first note that a theorem in $\T$ is a special well-formed formula. The bound we give is valid for the set of all the well-formed formulae. We
consider the following program $C$: on its input $x$, $C$ tests if $x$ is a well-formed formula. It outputs it if the case arises, and enters
in an infinite loop else. 

This program has to be modified a bit as its domain is not prefix-free. The idea here is to add at the end of the input
a marker which appears only at the end of the words. In that way, if $x$ is prefix of $y$, then the end-marker has to appear in $y$. As it can only
appear at the end of $y$, then $x=y$. It ensures that the domain is prefix-free. We now have to define an end-marker. It is sufficient to take an
ill-formed formula. More precisely, we need a formula $y$ such that for every well-formed formula $x$, $xy$ is ill-formed, and for every 
$z\in X^*_i$, $xyz$ is also ill-formed. For instance, we can take $y=++$, where the symbol $+$ is interpreted as the addition of natural numbers.
There are in all formal systems plenty of possibilities for this $y$ (another choice could be $(+$ for instance, or any ill-formed formula with
parenthesis around). In the sequel, $y$ represents a fixed such ill-formula.

The new machine $C$ works as follows: on an input $z$, $C$ checks if $z=xy$ with a certain $x$. If the case arises, it checks if $x$ is
a well-formed formula, and then outputs $x$ if it does. In all the other cases, $C$ diverges. Now, we have a new machine $C$ whose domain 
is prefix-free, and such that $H_C(x)\le\abs x_i+\abs{y}_i$. By the Invariance Theorem, we get a constant $c$ such that $H_i(x)\le\abs x_i +c$.

We now prove the lower bound, that is that the complexity of a theorem has to be greater than a half of its length, up to a constant. 
The idea is the following: If
we consider a sentence $x$ of the set of theorems $\T$, then it may contain some variables which cannot be compressed. More precisely, as we can
work with many variables, it is not possible that for each of these variable, the word which is used to represent it has a small complexity.
To formalize the idea, we have
to define in a formal way what the variables in our formal language are. We consider that the variables are created as follows. A variable is denoted
by a special character, say $v$, indicating that it is a variable, and then a binary-written number identifying each variable. This number is called
the identifier of the variable. In the sequel, we denote by $v_n$ the variable the identifier of which is the integer $n$.

Now, we have to consider the formulae defined by
\[\phi(m,n)\equiv\exists v_m \exists v_n (v_m=v_n).\]
We suppose that $m$ and $n$ are random strings, that is $H_i(m)\ge\abs m_i+\O(1)$ and $H_i(n)\ge\abs n_i+\O(1)$. Furthermore, we suppose that
$H(m,n)\ge\abs m_i+\abs n_i+\O(1)$, in other words that $m$ and $n$ together are random. We can suppose that as such words do exist. Then 
\begin{eqnarray*}
H_i(\phi(m,n))&\ge& H_i(m)+H_i(n)+\O(1)\\
&\ge&\abs m_i+\abs n_i+\O(1)\\
&\ge&\frac12\cdot\abs{\phi(m,n)}_i+\O(1).
\end{eqnarray*}
Thus, we obtained the lower bound.

\end{proof} 

Improving the bounds in this lemma seems to be hard. A preliminary work should be to define exactly what we accept as a formal language.

The next theorem is the formal version of Chaitin's ``heuristic principle''. The very substance of the proof comes from previous results.

\begin{theorem}[\protect{\cite[Theorem 4.6]{csi05}}]\label{thm-incompl} 
Consider a finitely-specified, arithmetically sound (i.e. each arithmetical proven sentence is true), consistent theory strong enough to formalize
arithmetic, and denote by $\T$ its set of theorems written in the alphabet $X_i$.  Let $g$ be a G\"odel numbering for $\T$. Then, there exists 
a constant $N$,  which depends upon $U_i, U_2$ and $\T$, such that $\T$ contains no $x$ with $\delta_g(x)>N$. 
\end{theorem} 

\begin{proof} 
By Lemma \ref{lem-bound-th}, for every $x\in\T$, $\delta_i(x)\le c$. Using Theorem \ref{thm-inva}, there exists a constant $N$ such that for
every $x\in\T$, $\delta_g(x)\le N$.
\end{proof} 

The $\delta_g$ measure is also useful to prove a probabilistic result about independent statements. Indeed, we can prove that the probability of a
true statement of length $n$ to be provable tends to zero when $n$ tends to infinity.

\begin{proposition}[\protect{\cite[Proposition 5.1]{csi05}}]\label{prop-lim} 
Let $N>0$ be a fixed integer,  $\T\subset X^*_i$ be c.e. and $g: \T \to B^*$ be a  G\"odel numbering. Then,
\begin{equation}
\lim_{n\to\infty} i^{-n}\cdot\card\ensemble{x\in X^*_i}{\abs x_i=n, \delta_g(x) \le N}=0.
\end{equation}
\end{proposition} 

We do not give a proof of this proposition because it is essentially technical. It can be found in \cite{csi05}. In Section \ref{sec-indep}, the
proof of Proposition \ref{prop-rho2} uses the same arguments and differs from this one only by details. Now, we can express the probabilistic
result about independent statements. The proof of this result can be found in \cite[p. 11]{csi05}.

\begin{theorem}[\protect{\cite[Theorem 5.2]{csi05}}]\label{thm-proba} 
Consider a consistent, sound, finitely-specified theory strong enough to formalize arithmetic. The probability that a true sentence of length $n$ is
provable in the theory tends to zero when $n$ tends to infinity.
\end{theorem} 

\section{Acceptable complexity measures}\label{sec-acm} 

The function $\delta_g$ is our model to build the notion of \emph{acceptable complexity measure of theorems}. At this end, we first define
what a \emph{builder} is, and then the properties it has to verify in order to be said \emph{acceptable}. An \emph{acceptable
complexity measure of theorems} will then be a complexity measure built \emph{via} an acceptable builder.

\begin{definition}\label{def-complexity} 
For a computable function $\hat\rho_i : \N\times\N\to\Q$, we define the \emph{complexity measure builder} $\rho$ by
\begin{eqnarray*}
\rho : G & \to & [X^*_i\to\Q]\\
        g & \mapsto & [u\mapsto\hat\rho_i(H_2(g(u)),\abs u_i)]
\end{eqnarray*}
The function $\hat\rho_i$ is called the \emph{witness} of the builder. In the sequel, we note $\rho_g(u)$ instead of $\rho(g)(u)$.
\end{definition} 

Now, we define three properties that a builder has to verify to be \emph{acceptable}. We recall that $\F$ denotes a theory which satisfy
the hypothesis of G\"odel Incompleteness Theorem, and $\T$ its set of theorems.

\begin{definition}\label{def-acceptable} 
A builder $\rho$ is said \emph{acceptable} if for every $g$, the measure $\rho_g$ verifies the three following conditions:
\begin{enumeration}
\item For every theory $\F$, there exists an integer $N_\F$ such that if $\F\vdash x$, then $\rho_g(x)<N_\F$.\label{cdt-bound}
\item For every integer $N$, \[\lim_{n\to\infty} i^{-n} \cdot \card\ensemble{x\in X^*_i}{\abs x_i=n\textand\rho_g(x)\le N} = 0.\]\label{cdt-lim}
\item For every G\"odel numbering $g'$, there exists a constant $c$ such that for every string $u\in X^*_i$, $\abs{\rho_g(u)-\rho_{g'}(u)}\le c$.
    \label{cdt-inva}
\end{enumeration}
\end{definition} 

The first property is simply the formal version of Chaitin's ``heuristic principle''. The second one corresponds to Proposition \ref{prop-lim} and
eliminate trivial measures. Finally, \cdt{inva} ensures the independence on the way the theorems are written.
In other words, the properties \cdt{bound}, \cdt{lim} and \cdt{inva} ensure that an acceptable complexity measure satisfy Theorem \ref{thm-incompl},
Proposition \ref{prop-lim} and Corollary \ref{cor-inva} respectively.

The following proposition will be useful in the sequel. It is a weaker version of the property \cdt{bound} which is used to prove that a
measure is not acceptable, and more precisely that it does not satisfy this first property.

\begin{proposition}\label{prop-inf} 
Let $\rho_g$ be an acceptable complexity measure. Then there exists an integer $N$ such that for every integer $M\ge N$, the set
\begin{equation}\label{eq-inf}
\ensemble{x\in X^*_i}{\rho_g(x)\le M}
\end{equation}
is infinite.
\end{proposition} 

\begin{proof} 
We consider a theory $\F$ and the integer $N_\F$ given by the property \cdt{bound} in Definition \ref{def-acceptable}. Clearly, $\F$ can prove an
infinity of theorems, such as ``$n=n$'' for all integer $n$. All of them
have by property \cdt{bound} a complexity bounded by $N_\F$. If $\T$ is the set of theorem that $\F$ proves, then
\[\T\subset\ensemble{x\in X^*_i}{\rho_g(x)\le N_\F}.\]
As $\T$ is infinite, so is the set in the proposition, and it remains true for every $M\ge N_\F$. 
\end{proof} 

We now prove that the $\delta_g$-complexity is an acceptable complexity measure. This result is natural as the notion of acceptable complexity 
measure was built to generalize $\delta_g$.

\begin{proposition}\label{prop-delta-acc} 
The function $\delta_g$ is an acceptable complexity measure.
\end{proposition} 

\begin{proof} 
The $\delta_g$ function we defined plays the role of $\rho_g$. We have to provide an acceptable builder. Let define
\[\hat\delta_i(x,y)=x-\ceil{\log_2(i)\cdot y}\]
which plays the role of $\hat\rho_i$. Then $\delta_g(x)=\hat\delta_i(H_2(g(x)),\abs x_i)$.

In fact, the properties of $\delta_g$ proved in \cite{csi05} are exactly what we need here.
One can easily check that
\cdt{bound} is ensured by Theorem \ref{thm-incompl},
\cdt{lim} by Proposition \ref{prop-lim} and \cdt{inva} by Corollary \ref{cor-inva}.
\end{proof} 

The goal of defining an acceptable builder and an acceptable measure is to study other complexities than $\delta_g$. The following example proves that
the program-size complexity is not acceptable. This result, even though it is plain, is very important. Indeed, it justifies the need to define other
complexity measures.

\begin{example} 
A first natural complexity to study is the program-size complexity. There is no difficulty in verifying that $H$ is a complexity measure. Formally, we
have to define $\hat\rho_i(x,y)=x$ and such that $H_2(g(x))=\hat\rho_i(x,\abs x_i)$. We study the properties of the builder 
$g\mapsto [x\mapsto H_2(g(x))]$.
Let us see how it behaves with the three properties of Definition \ref{def-acceptable}.
\begin{enumeration}
\item This first property cannot be verified. Indeed, we note that 
    \begin{eqnarray*}
    &&\card\ensemble{x\in X^*_i}{H_2(g(x))\le N}\\
    &\le&\card\ensemble{y\in X_2^*}{H_2(y)\le N}\\
    &\le& 2^N.
    \end{eqnarray*}
    If the property was verified, the set of theorems $\T$ proved by $\F$ would be bounded by $2^N$, a contradiction.
\item This property is on the contrary obviously verified. Indeed, as $\card\ensemble{x\in X^*_i}{H_2(g(x))\le N}\le 2^N$, 
    $\ensemble{x\in X^*_i}{\abs x_i=n\textand H_2(g(x))\le N}=\emptyset$ for large enough $n$.
\item This property corresponds exactly to Corollary \ref{cor-inva}, and is verified.
\end{enumeration}
\end{example} 

As the program-size complexity cannot be used there, we try to find other complexities which better reflect the intrinsic
complexity. That is why we use the length of the strings to alter the complexity. It seems natural that the longest strings are also the most
difficult to describe\footnote{One has to be very careful with this statement which is not really true.}. In the next section, we will give two
other examples of builder which are not acceptable.

\section{Independence of the three conditions}\label{sec-indep} 

The aim of this section is to prove that the conditions \cdt{bound}, \cdt{lim} and \cdt{inva} in Definition \ref{def-acceptable} are
independent from each other. At this end, we give two new examples of unacceptable builders. Each of those unacceptable builders exactly satisfy
two conditions in Definition \ref{def-acceptable}. Furthermore, they give us a first idea of the ingredients needed to build an acceptable
complexity builder. In particular they show us that a builder shall neither be too small nor too big.

\begin{example}\label{ex-rho1} 
Let $\hat\rho^1_i$ be the function defined by $\hat\rho^1_i(x,y)=x/y$ if $y\neq0$ and $0$ else. It defines a builder $\rho^1$ and for every G\"odel
numbering $g$, we can define $\rho^1_g$ by
\[\rho^1_g(x)= \begin{cases}
    \frac{H_2(g(x))}{\abs x_i}, & \textif x\neq\lambda,\\
    0, & \textelse.
\end{cases}\]
\end{example} 

We will see in the sequel that $\rho^1$ is a too small complexity. In fact, it is even bounded. In order to avoid this problem, we define 
$\rho^2$ by dividing the program-size complexity by the logarithm of the length.

\begin{example}\label{ex-rho2} 
We consider $\hat\rho^2_i$ defined by
\[\hat\rho^2_i(x,y)= \begin{cases}
    \frac{x}{\ceillog{i}{y}}, & \textif y>1,\\
    0, & \textelse.    
\end{cases}\]
The corresponding builder applied with a G\"odel numbering $g$ defines the function
\[\rho^2_g(x)= \begin{cases}
    \frac{H_2(g(x))}{\ceillog{i}{\abs x_i}}, & \textif\abs x_i>1,\\
    0, & \textelse.
\end{cases}\]
\end{example} 

In order to make the proofs easier, we introduce a new function for each already defined builders. Those functions make no use of 
G\"odel numberings. They are the equivalents of $\delta_i$ for $\rho^1$ and $\rho^2$.
They can help us in the proofs because we prove first that they are up to a constant equal to the complexity measures.
For $\rho^1$, we define $\rho^1_i$ be by $\rho^1_i(x)=H_i(x)/\abs x_i$ if $x\neq\lambda$ and $0$ else. And similarly, for $\rho^2$, we define
$\rho^2_i(x)=H_i(x)/\ceillog{i}{\abs x_i}$ if $\abs x_i>1$ and $0$ else.

\begin{lemma}\label{lem-rho-inva} 
Let $A\subseteq X^*_i$ be c.e. and $g:A\to B^*$ be a G\"odel numbering. Then, there effectively exists a constant $c$ (depending upon $U_i$, $U_2$ and $g$)
such that for all $u\in A$, we have
\begin{equation}\label{eq-rhoi-inva}
\abs{\rho^j_g(u)-\log_2(i)\cdot\rho^j_i(u)}\leq c,
\end{equation}
$j=1,2$.
\end{lemma} 

\begin{proof} 
We first note that this difference is null for $u=\lambda$ in the case $j=1$, and for $\abs u_i\le1$ in the case $j=2$. In the sequel, we suppose
that $\abs u_i>0$ (for $j=1$) or $\abs u_i>1$ (for $j=2$).

Theorem \ref{thm-inva} states that
\[\abs{H_2(g(u))-\log_2(i)\cdot H_i(u)}\leq c.\]
We now just have to divide the whole inequality by $\abs{u}_i\geq1$ to obtain (\ref{eq-rhoi-inva}) with $j=1$ and by $\ceillog{i}{\abs u_i}$
which is not less than one but for finitely many $u$ to obtain the result with $j=2$.
\end{proof} 

This result allows us to work with much easier forms of the complexity functions. 
We now study the properties that $\rho^1_g$ and $\rho^2_g$ satisfy.
As a corollary of the above lemma, we can note that both of the measures satisfy \cdt{inva}.

\begin{proposition}\label{prop-rho1} 
The function $\rho^1_g$ verifies condition \cdt{bound} in Definition \ref{def-acceptable}, but does not verify \cdt{lim}.
\end{proposition} 

\begin{lemma}\label{lem-rho1-bounded} 
There exists a constant $M$ such that for all $x\in X^*_i$, $\rho^1_g(x)\leq M$.
\end{lemma} 

\begin{proof} 
The result is plain for $x=\lambda$. We now suppose that $\abs x_i>0$.
In view of \cite[Theorem 3.22]{irap02}, there exist two constants $\alpha$ and $\beta$ such that for all $x\in X^*_i$,
\[H_i(x)\leq \abs x_i + \alpha\cdot\log_i\abs x_i + \beta,\]
so, for $x\neq\lambda$,
\[\rho^1_i(x)\leq 1 + \alpha\cdot\frac{\log_i\abs x_i}{\abs x_i} + \beta\cdot\frac{1}{\abs x_i}\cdot\]
As $\log_i(\abs x_i)/\abs x_i\leq 1$ for every $x\neq\lambda$, then 
\[\rho^1_i(x)\leq 1 + \alpha+\beta.\]
Furthermore, Lemma \ref{lem-rho-inva} states that for every $x$, we have
\begin{eqnarray*}
\rho_g^1(x) & \le & c + \log_2(i)\cdot\rho^1_i(x)\\
& \le & c+\log_2(i)\cdot(1+\alpha+\beta).
\end{eqnarray*}
Accordingly, $M=\ceil{c+\log_2(i)\cdot(1+\alpha+\beta)}$ satisfies the statement of the lemma.
\end{proof} 

\begin{proof}[Proof of Proposition \ref{prop-rho1}] 
The property \cdt{bound} is obvious since Lemma \ref{lem-rho1-bounded} tells us that the bound is valid for every sentence $x$, 
not only provable ones. On the contrary, the fact that
$\rho^1_g$ is bounded by $M$ implies that for $N\ge M$, the set $\ensemble{x\in X^*_i}{\abs x_i=n\textand\rho^1_g(x)\le N}$ is the set
$X^n_i$. Hence the limit of \cdt{lim} is $1$ instead of $0$.
\end{proof} 

The above proof shows us that an acceptable complexity measure cannot be too small ($\rho^1$ is even bounded). We will now see, thanks to the complexity
measure $\rho^2$, that an acceptable complexity measure cannot be too big either.

\begin{proposition}\label{prop-rho2} 
The function $\rho^2_g$ verifies condition \cdt{lim} in Definition \ref{def-acceptable}, but does not verify \cdt{bound}.
\end{proposition} 

\begin{proof} 
We begin with the proof of \cdt{lim} for $\rho^2$. 
Theorem \ref{lem-rho-inva} allows us to consider $\rho^2_i$ instead of $\rho^2_g$, with a new constant $\ceil{(N+c)/\log_2(i)}$. 
Indeed, it states that $\rho^2_g(x)\ge\log_2(i)\cdot\rho^2_i(x)-c$, and consequently
\[\ensemble{x\in X^n_i}{\rho^2_g(x)\le N}\subseteq\ensemble{x\in X^n_i}{\rho^2_i\le\ceil{\frac{N+c}{\log_2(i)}} }.\]
In order to avoid too many notations, we still denote this constant by $N$.

First, we note that
\[\ensemble{x\in X^n_i}{\rho^2_i(x)\leq N}=\ensemble{x\in X^n_i}{\exists\,y\in X^{\leq N\cdot\ceillog{i}{n}}_i,\ U_i(y)=x}.\]
Translating in terms of cardinals, we obtain
\begin{eqnarray*}
&&\card\ensemble{x\in X^n_i}{\rho^2_i(x)\leq N}\\
 & \le & \card\ensemble{x\in X^n_i}{\exists\,y\in X^{\leq N\cdot\ceillog{i}{n}}_i,\ U_i(y)=x}\\
 & \le & \card\ensemble{y\in X^{\leq N\cdot\ceillog{i}{n}}_i}{\abs{U_i(y)}=n}\\
 & \le & \card\ensemble{y\in X^{\leq N\cdot\ceillog{i}{n}}_i}{U_i(y)\text{ halts.}}\\
 & \le & \sum_{k=1}^{N\cdot\ceillog{i}{n}}\underbrace{\card\ensemble{y\in X^k_i}{U_i(y)\text{ halts.}} }_{r_k}
\end{eqnarray*}

We extend these inequalities to the limit when $n$ tends to infinity:
\begin{eqnarray*}
&&\lim_{n\to\infty} i^{-n}\cdot\card\ensemble{x\in X^n_i}{\rho^2_g(x)\leq N}\\
 & \le & \lim_{n\to\infty}\sum_{k=1}^{N\cdot\ceillog{i}{n}} i^{-n}\cdot r_k \\
 & \le & \lim_{n\to\infty}i^{N\cdot\ceillog{i}{n}-n}\cdot\sum_{k=1}^{N\cdot\ceillog{i}{n}} i^{-N\cdot\ceillog{i}{n}}\cdot r_k.
\end{eqnarray*}

We note that 
\[\lim_{n\to\infty}\sum_{k=1}^{N\cdot\ceillog{i}{n}} i^{-N\cdot\ceillog{i}{n}}\cdot r_k=\lim_{m\to\infty}\sum_{k=1}^m i^{-m}\cdot r_k.\]
Now, 
\[\lim_{m\to\infty}\frac{\displaystyle\sum_{k=1}^{m+1} r_k - \sum_{k=1}^m r_k}{i^{m+1}-i^m} = \frac i{i-1}\cdot\lim_{m\to\infty}i^{-m}\cdot r_m=0.\]
The last inequality comes from Kraft's inequality:
\[\sum_{m=1}^\infty i^{-m}\cdot r_m\le 1.\]
So we can apply Stolz-Cesàro Theorem to ensure that
\begin{equation}\label{eq-prob1}
\lim_{n\to\infty}\sum_{k=1}^{N\cdot\ceillog{i}{n}} i^{-N\cdot\ceillog{i}{n}}\cdot r_k=0.
\end{equation}

On the other hand, 
\begin{equation}\label{eq-prob2}
\lim_{n\to\infty}i^{N\cdot\ceillog{i}{n}-n}=0.
\end{equation}

We just have to combine (\ref{eq-prob1}) and (\ref{eq-prob2}) to obtain \cdt{lim}.

Now, it remains to prove that \cdt{bound} is not verified. At this end, we suppose that \cdt{bound} holds. 
We note $\T$ the set of theorems that $\F$ proves. Note first that
\begin{align}
&\card\ensemble{x\in X^*_i}{\abs x_i=n\textand H_2(g(x))\leq N\cdot\ceillog{i}{n}}\\
 & \le \card\ensemble{y\in B^*}{H_2(y)\le N\cdot\ceillog{i}{n}}\notag\\
 & \le 2^{N\cdot\ceillog{i}{n}} \notag\\
 & \le 2^{N\cdot(\log_i n+1)} \notag\\
 & \le 2^N\cdot n^{N\cdot\log_i 2}.\label{eq-poly}
\end{align}
So, if \cdt{bound} holds for all $x\in\T$, we have
\begin{equation}\label{eq-card-T}
\card\ensemble{x\in\T}{\abs x=n}\leq \alpha n^{\beta N},
\end{equation}
for every integer $n$, where $\alpha$ and $\beta$ come from (\ref{eq-poly}).

But we now consider the set of formulae
\[\Phi_k = \ensemble{Q_0x_0Q_1x_1\dots Q_kx_k\ \bigwedge_{l=0}^k (x_l=x_l)}{Q_l\in\{\forall,\exists\}}.\]
Each formula $\phi\in\Phi_k$ is true, and all formulae have the same length $n_k=\O(k)$. Furthermore, $\card\ \Phi_k=2^k$.

As all those formulae belong to the predicate logic, all of them are provable in $\F$, that is to say they belong to $\T$. As we can take $k$ as
big as wanted, we can also have $n_k$ as big as wanted. 

Now we have, for arbitrary large $n$, $2^{\O(n)}$ formulae of length $n$ which belong to $\T$. That contradicts (\ref{eq-card-T}),
and so, \cdt{bound} is false.
\end{proof} 

We can now prove that \cdt{bound}, \cdt{lim} and \cdt{inva} in Definition \ref{def-acceptable} are independent from each other. As we know, with 
$\delta_g$, that there exists an acceptable complexity builder, it is sufficient to prove that for each of the three conditions, there exists a builder
which does not satisfy it while it satisfies both other ones.

\begin{theorem}\label{thm-independence} 
Each condition in Definition \ref{def-acceptable} is independent from others.
\end{theorem} 

\begin{proof} 
The measure builder $\rho^1$ is an measure example which satisfies both \cdt{bound} and \cdt{inva} but not \cdt{lim} while $\rho^2$ does not satisfy
\cdt{bound} but \cdt{lim} and \cdt{inva}. To prove the complete independence of the three conditions, it remains to prove that a complexity measure builder
can satisfy both \cdt{bound} and \cdt{lim} without satisfying \cdt{inva}.

In fact, our proof here does not exactly follow the scheme we gave. It is still unknown if all the complexity measure builders satisfy \cdt{inva}, or if
there exist some of them not satisfying it. Thus, the proof is built as follows. We prove that either all complexity builders satisfy \cdt{inva}, or there
exists at least one complexity builder satisfying \cdt{bound} and \cdt{lim} without satisfying \cdt{inva}. We also give the exact question the answer of
which would make the choice between the both possibilities.

Let $g$ and $g'$ be two G\"odel numberings from $X^*_i$ to $X^*_2$, and $\rho_g$ and $\rho_{g'}$ two complexity measures built with the same builder. 
The question is to know if $H_2(g(x))=H_2(g'(x))$ for all but finitely many $x\in X^*_i$
or if there exists an infinite sequence $(x_n)_{n\in\N}$ such that $H_2(g(x_n))\neq H_2(g'(x_n))$ for all $n$. Suppose that the first case holds, then
for all but finitely many $x\in X^*_i$, $\rho_g(x)=\hat\rho_i(H_2(g(x)),\abs x_i)=\hat\rho_i(H_2(g'(x)),\abs x_i)=\rho_{g'}(x)$. Consequently 
\[c=\max\ensemble{\abs{H_2(g(x))-H_2(g'(x))}}{x\in X^*_i}<\infty,\]
and the builder $\rho$ satisfy \cdt{inva}.

We suppose now that the second case holds, that means that there exist infinitely many strings $x\in X^*_i$ such that $H_2(g(x))\neq H_2(g'(x))$. 
We consider the acceptable complexity measure $\delta_g$. We define the measure $\rho_g$ by $x\mapsto\delta_g(x)^2$.  More formally, if we denote by
$\hat\delta_i$ the witness of the builder $\delta$, we define the builder $\rho$ \emph{via} the witness $\hat\rho_i=\hat\delta_i^2$. 
Let us consider the behaviour of this function with the three properties:

\begin{enumeration}
\item As $\delta_g$ is acceptable, there exists $N_\F$ such that if $\F\vdash x$, then $\delta_g(x)\le N_\F$. Then it
is plain that $\rho_g(x)\le N_\F^2$. So (i) is verified.
\item For an integer $N\ge1$, if $\rho_g(x)\le N$, then $\delta_g(x)\le N$ too. So we have the following:
    \begin{eqnarray*}
    &&\ensemble{x\in X^*i}{\abs x_i=n\textand\rho_g(x)\le N}\\&\subset&\ensemble{x\in X^*_i}{\abs x_i=n\textand\delta_g(x)\le N}.\end{eqnarray*}
    Consequently,
    \begin{eqnarray*}
        \lefteqn{\lim_{n\to\infty}i^{-n}\cdot\card\ensemble{x\in X^*i}{\abs x_i=n\textand\rho_g(x)\le N}}\\
        &\le&\lim_{n\to\infty}i^{-n}\cdot\card\ensemble{x\in X^*_i}{\abs x_i=n\textand\delta_g(x)\le N}=0.
    \end{eqnarray*}
    So (ii) is also verified.
\item We first note that 
    {\setlength\arraycolsep{2pt}
    \begin{eqnarray*}
    &&\rho_g(x)-\rho_{g'}(x)\\
    &=&\delta_g(x)^2-\delta_{g'}(x)^2\\
    &=&(H_2(g(x))-\ceil{\log_2(i)\cdot\abs x_i})^2\\
    &&-(H_2(g'(x))-\ceil{\log_2(i)\cdot\abs x_i})^2\\
    &=&(H_2(g(x))^2-H_2(g'(x))^2)\\
    &&-2\cdot\ceil{\log_2(i)\cdot\abs x_i}(H_2(g(x))-H_2(g'(x))).
    \end{eqnarray*}}
    We know from Corollary \ref{cor-inva} that $(H_2(g(x))-H_2(g'(x)))$ is bounded. Thus, we only need to prove that $\abs{H_2(g(x))^2-H_2(g'(x))^2}$ is
    unbounded, and we will be able to conclude that \cdt{inva} is not satisfied by $\rho$. Suppose that it is bounded by an integer $N$. As we have supposed
    that there exist infinitely many $x\in X^*_i$ such that $H_2(g(x))\neq H_2(g'(x))$, then there exists for every integer $M$ a string $x$ such that 
    $H_2(g(x))>H_2(g'(x))>M$\footnote{We can impose here without any loss of generality that $H_2(g(x))>H_2(g'(x))$ because the converse situation would
    be equivalent.}. Then
    \begin{eqnarray*}
    &&H_2(g(x))^2-H_2(g'(x))^2\\&=&(H_2(g(x))-H_2(g'(x)))\cdot(H_2(g(x))+H_2(g'(x)))\\
    &>&1\cdot(2\cdot M)=2M.
    \end{eqnarray*}
    We can also conclude, using an integer $M>N/2$ that this bound cannot exist, that is \cdt{inva} is not satisfied.
\end{enumeration}
\end{proof} 


\section{Form of the acceptable complexity measures}\label{sec-form} 

The aim of this section is to give some conditions that a complexity measure has to verify to be acceptable. More precisely, we will study some
conditions a builder, and in particular its witness, has to verify such that the complexity measures it builds are acceptable ones. We restrict
our study to particular witnesses, such as linear functions in both variables, or functions defined by
\[\hat\rho_i(x,y)=\frac{x}{f(y)}\]
where $f$ is a computable function.

Our first result shows a kind of stability of the acceptable complexity measures. Furthermore, it makes the following proofs easier.

\begin{proposition}\label{prop-stab} 
Let $\rho_g$ be an acceptable complexity measure, and $\alpha,\beta\in\Q$ such that $\alpha>0$. Then $\alpha\cdot\rho_g+\beta$ is also
an acceptable complexity measure.
\end{proposition} 

\begin{proof} 
Property \cdt{bound} in Definition \ref{def-acceptable} remains true with a new constant $\alpha\cdot N+\beta$ instead of $N$. In the same
way, \begin{eqnarray*}
        &&\ensemble{x\in X^*_i}{\abs x_i=n\textand\alpha\cdot\rho_g(x)+\beta\le N}\\
        &\subseteq&
        \ensemble{x\in X^*_i}{\abs x_i=n\textand\rho_g(x)\le \ceil{\frac{N-\beta}{\alpha}} },\end{eqnarray*}
hence Property \cdt{lim} is verified. Now, if we consider two G\"odel numberings $g$ and $g'$,
\[\abs{(\alpha\cdot\rho_g(x)+\beta)-(\alpha\cdot\rho_{g'}(x)+\beta)}=\alpha\cdot\abs{\rho_g(x)-\rho_{g'}(x)}\le\alpha\cdot c,\]
which proves that Property \cdt{inva} is retained.
\end{proof} 

We start studying the linear in both variables witnesses. The result we obtain is partial. However, as discussed after Lemma \ref{lem-bound-th},
this result is not likely to be improved without a complete study of the definition of the formal languages.

\begin{proposition}\label{prop-linear} 
Let $f$ be a function of two variables, linear in both variables such that $\hat\rho_i$ defined by $\hat\rho_i(x)=\floor{f(x)}$ is computable. 
If $\hat\rho_i$  defines an acceptable complexity measure, then there exist $a, b$ and $\epsilon$, $a>0$ and $1/2\le\epsilon\le1$, such that
\[\hat\rho_i(x,y) = \floor{a\cdot(x-\epsilon\cdot\log_2(i)\cdot y)+b}.\]
\end{proposition} 

\begin{proof} 
We consider any function
which satisfies the hypothesis. Then there exist $\alpha, \beta$ and $\gamma$ such that 
\[\hat\rho_i(x,y)=\floor{\alpha x-\beta y + \gamma xy}.\]
Proposition \ref{prop-stab} allows us to fix $\hat\rho_i(0,0)=0$.
Of course, it would be equivalent to consider $\alpha x+\beta y+\gamma xy$, but the chosen version simplifies the notations. 
Let $\beta'$ be such that $\beta=\beta^\prime\cdot\log_2(i)$.
The proof is done
in several steps. We start by showing that one at least of $\alpha$ and $\gamma$ has to be different from zero, then that $\gamma=0$. After that, we
prove that $\alpha/2\le\beta^\prime\le\alpha$. 

Suppose that $\alpha=\gamma=0$. Then $\rho_g(x)=-\ceil{\beta\abs x_i}$. If $\beta\le0$, then Proposition \ref{prop-inf} is not verified by our complexity
measure, and hence neither is Property \cdt{bound}. If $\beta\ge0$, it is obvious that Property \cdt{lim} cannot hold true.

Then, we use the property \cdt{bound} and consider the set
\begin{eqnarray*}
    &&\ensemble{x\in X^*_i}{\abs x_i=n\textand\rho_g(x)\le N}\\
    &\subseteq&
    \ensemble{x\in X^*_i}{\abs x_i=n\textand H_2(g(x))\le\ceil{\frac{\beta n+N+1}{\gamma n+\alpha}} }.\end{eqnarray*}
Furthermore,
\[\lim_{n\to\infty}\frac{\beta n+N+1}{\gamma n+\alpha}=\left\{ \begin{array}{ll}
    \beta/\gamma, & \textif\gamma\neq0;\\
    (N+1)/\alpha, & \textif\gamma=\beta=0;\\
    \pm\infty, & \textif\gamma=0\textand\beta\neq0.
    \end{array}\right.\]
The only solution is the third one because in order to satisfy \cdt{bound}, this limit has to be infinite. Indeed, if it is finite, we can use the
same proof as in Proposition \ref{prop-rho2} to conclude to a contradiction. So we know that $\gamma=0$, and hence that $\alpha\neq0$. We can right
now say that $\alpha$ and $\beta$ have the same sign, because the limit cannot be $-\infty$. Using Proposition \ref{prop-stab}, we can assume that
$\alpha=1$. Indeed, $\alpha<0$ is not possible because of Property \cdt{lim}.

To make easier the remaining of the proof, we define an auxiliary measure as we did in Sections \ref{sec-delta} and \ref{sec-indep} 
for $\delta$, $\rho^1$ and $\rho^2$. Let $\rho_i$ be defined by
\[\rho_i(x)=\floor{H_i(x)-\beta^\prime\cdot\abs x_i}.\]
Applying Theorem \ref{thm-inva}, we get a constant $c$ such that for every $x$,
\[\abs{\rho_g(x)-\log_2(i)\cdot\rho_i(x)}\le c.\]

We will now use the property \cdt{lim} to have other information on $\beta^\prime$, and hence $\beta$. We only know at that stage that $\beta^\prime>0$. 
We consider the set
\begin{eqnarray*}
    &&\ensemble{x\in X^*_i}{\abs x_i=n\textand\rho_g(x)\le N}\\
    &\subseteq&
    \ensemble{x\in X^*_i}{\abs x_i=n\textand H_i(x)\le\beta^\prime\cdot n+N+c+1}.
\end{eqnarray*}
If $\beta^\prime>1$, then for every constant $d$, if we choose $n$ large enough we have $\beta^\prime\cdot n>n+d\cdot\log n$. And we can use
the inequality $H_i(x)\le\abs x_i+\O(\log_i\abs x_i)$ (see \cite[Theorem 3.22]{irap02}) to conclude that the above set is $X^n_i$. 
And so, property \cdt{lim} is not verified, the limit being $1$.

Using now the lower bound in Lemma \ref{lem-bound-th}, we know that for every proven sentence $x$,
\[H_i(x)\ge \frac12\cdot\abs x_i.\]
Suppose that $\beta^\prime<1/2$. Then for every $x$ such that $\F\vdash x$,
\[\rho_i(x)=\left(H_i(x)-\frac12\cdot\abs x_i\right)+(\frac12-\beta^\prime)\cdot\abs x_i\ge(\frac12-\beta^\prime)\cdot\abs x_i.\]
Thus, \cdt{bound} cannot be verified.
\end{proof} 

We study another kind of witnesses. Functions defined by 
\[\hat\rho_i(x,y)=\frac{x}{f(y)}\]
where $f$ is a computable function may be interesting because they are the only reasonable candidates for being witness of \emph{multiplicative} 
complexity measures. Indeed, a complexity of the form $H_2(g(x))\cdot\abs x_i$ has no chance to satisfy the desired properties. Unfortunately, such
functions never define acceptable measures. 

\begin{proposition}\label{prop-div} 
Let $f$ be a computable function, and $\hat\rho_i$ defined by
\[\hat\rho_i(x,y)=\frac{x}{f(y)}\cdot\]
Then the complexity measure builder the witness of which is $\hat\rho_i$ cannot satisfy at the same time properties \cdt{bound} and \cdt{lim}.
\end{proposition} 

\begin{proof} 
Suppose that $\rho_g(x)=\hat\rho_i(H_2(g(x)),\abs x_i)$ satisfy \cdt{bound}. Then consider the set 
\[\ensemble{x\in X^*}{\abs x_i=n\textand H_2(g(x))\le N\cdot f(n)}.\]
Its cardinal is at most $2^{N\cdot f(n)}$. Furthermore, this set contains the set
of all the sentences in $\T$ the length of which is $n$. Hence,
\begin{equation}\label{eq-low1}
\card\ensemble{x\in\T}{\abs x_i = n}\leq 2^{N\cdot f(n)}.
\end{equation}

Now, we give a lower bound to this cardinal. The proof of Proposition \ref{prop-rho2} shows that this cardinal is greater to $2^{\O(n)}$.
Accordingly, there exists a constant $c$ such that
\begin{equation}\label{eq-low2}
\card\ensemble{x\in\T}{\abs x_i = n}\geq 2^{c\cdot n}.
\end{equation}

We also obtain that $2^{c\cdot n}\le 2^{N\cdot f(n)}$. We can conclude that
\begin{equation}\label{eq-f-too-big}
f(n)\ge\frac cN\cdot n.
\end{equation}

We now follow the proof we made to show that $\rho^1_g$ does not satisfy \cdt{lim}. We can define
\[\rho_i(x)=\frac{H_i(x)}{f(\abs x_i)}\raisebox{0.7pt}{,}\]
and we prove as for $\rho^1$ and $\rho^2$ that there exists a constant $d$ such that
\[\abs{\rho_g(x)-\log_2(i)\cdot\rho_i(x)}\leq d.\]
The proof of Lemma \ref{lem-rho-inva} is still valid here. In the same way, we extend Lemma \ref{lem-rho1-bounded} to $\rho_g$, namely there
exists a constant $M$ such that $\rho_g$ is bounded by $M$. Considering $\rho_g$ instead of $\rho^1_g$ has just an influence on the value
of the constant $M$. 

Now, we have to note that for $N\ge M$, the set $\ensemble{x\in X^*_i}{\abs x_i=n\textand\rho_g(x)\le N}$ is the set $X^n_i$ to conclude that 
property \cdt{lim} is not verified.
\end{proof} 


\section{Concluding remarks}

In this paper, we have studied the $\delta_g$ complexity function defined by Calude and J\"urgensen \cite{csi05}. This study
has led us to modify a bit the definition of $\delta_g$ in order to correct some of the proofs. Then, we have been able to
propose a definition of \emph{acceptable complexity  measure of theorem} which captures the main properties of $\delta_g$.
Studying some complexity measures, we have shown that the conditions of acceptability are quite hard to complete. Yet,
the definition seems to be robust enough to allow some investigations to find other natural acceptable complexity measures.

There remain some open questions. Among them, we can express the following ones:
\begin{itemize}
\item Can we improve the bounds of Lemma \ref{lem-bound-th}? This question could be interesting not only to improve Proposition
    \ref{prop-linear} but also for itself: How simple are the well-formed formulae, and in other words, to what extent can we use
    their great regularities to compress them? Yet, as already discussed, this question needs to be better defined. In particular,
    one has to investigate about the definition of the formal languages. The answer seems to be very dependent on the considered
    language.
\item Do there exist some acceptable complexity measure which are very different from $\delta_g$? The idea here is to find some
    measures with which we go further on the investigations about the roots of unprovability.
\item In view of the proof of Theorem \ref{thm-independence}, if we have two G\"odel numberings $g$ and $g'$, does the equality $H_2(g(x))=H_2(g'(x))$
    hold for all but finitely many $x$ or are those two quantities infinitely often different from each other? 
\end{itemize}

Those few questions are added to the ones Calude and J\"urgensen expressed in \cite{csi05}. The goal of finding new acceptable
complexity measures is to have new tools to try to answer their questions, as the existence of independent sentences of small
complexity.

\section*{Acknowledgments} 

Special thanks are due to Cristian Calude without whom these paper would have never existed. His very helpful comments, corrections and
improvements, as well as his hospitality made my stay in Auckland much nicer than all what I could hope. Thanks are also due to Andr\'e Nies 
for his comments and ideas. In particular, he gave us the lower bound in Lemma \ref{lem-bound-th}. 


\end{document}